\documentclass[a4paper,twoside]{article}

\usepackage{epsfig}
\usepackage{url}
\usepackage{comment}
\usepackage{subcaption}
\usepackage{calc}
\usepackage{amssymb}
\usepackage{amstext}
\usepackage{amsmath}
\usepackage{amsthm}
\usepackage{multicol}
\usepackage{pslatex}
\usepackage{apalike}
\usepackage{algorithm2e}
\usepackage[bottom]{footmisc}
\usepackage{amsmath,amssymb,amsthm}
\usepackage{graphicx}
\usepackage[colorlinks=true,citecolor=blue,linkcolor=blue,urlcolor=blue]{hyperref}
\usepackage{cite}
\usepackage[shortlabels]{enumitem}
\usepackage{booktabs}
\usepackage{tabularx}
\usepackage{xcolor}
\usepackage{algorithm2e}
\usepackage{tikz}
\usetikzlibrary{shapes,arrows,positioning,calc,fit}
\usepackage[T1]{fontenc}
\usepackage{times}
\usepackage{float}
\usepackage{pgfplots}
\pgfplotsset{compat=1.17}

\newtheorem{theorem}{Theorem}

\newtheorem{definition}{Definition}

\newtheorem{remark}{Remark}

\usepackage{SCITEPRESS}     

\begin{document}

\title{upTPM: Unbounded Preprocessing for Schnorr Multi-Signatures on TPM}

\author{\authorname{Yunusa Simpa Abdulsalam\sup{1}, Mustapha Hedabou\sup{1}}
\affiliation{\sup{1}College of Computing, University Mohammed VI Polytechnic, Benguerir, Morocco}
\email{abdulsalam.yunusa@um6p.ma}
}


\abstract{Schnorr-based multi-signature schemes support offline preprocessing of nonce commitments to reduce online signing to a single round. However, preprocessing is inherently bounded: each preprocessed nonce pair consumes signer-side storage, and once exhausted, an interactive commitment round is required to refill. This limitation is particularly severe for TPM~2.0 devices, where usable NVRAM is typically 6--16\,KB and connectivity is intermittent. This paper presents upTPM, a framework that achieves unbounded preprocessing with constant signer storage. Each TPM stores a single 32-byte secret seed from which an unlimited sequence of nonce commitments is deterministically derived. Commitments are published to an untrusted coordinator before use; nonce scalars never leave the TPM. We formalize three properties not provided by existing schemes: (1)~unbounded deterministic preprocessing with constant storage; (2)~asynchronous commitment refill, allowing any signer to unilaterally extend its commitment pool; and (3)~TPM-attested commitments, a hardware-backed authenticity and state-binding mechanism that strengthens resistance to host-software compromise. We prove EU-CMA security in the random oracle model under the discrete logarithm assumption and Pseudo Random Function (PRF) security, with a one-time-use invariant enforced by TPM hardware state. We extend the construction to $(t,n)$-threshold signatures and provide a detailed analysis of coordinator trust, crash recovery, and performance evaluations.}

\onecolumn \maketitle \normalsize \setcounter{footnote}{0} \vfill

\section{\uppercase{Introduction}}
\label{sec:introduction}

The modern internet era has enhanced the easy execution of digital applications in finance and IT, with electronic devices being the drivers for their implementation and usage. The need for electronic devices in today's world has become crucial. Research shows that 80.61\% of the 2021 world population owned a mobile device, which is a 31.21\% increase from 2016, and it is expected to reach 91\% by the year 2025 \cite{Statista}. In developed countries like the United States, Singapore, and Australia, an average user is likely to own more than one digital electronic device, with over 80\% equipped with TPM chips \cite{chen2014} . The literature translates users' ability to own more than one electronic device as a possibility for cross-device functionality ranging from financial operations to photos and video synchronization \cite{svenda2024}. A typical example is how electronic devices offer digital wallet implementation for payment. Unfortunately, most financial institutions have had difficulty adopting digital wallets in the financial sector, owing to their unsecured design or failed implementations. 

Cryptographic signatures and algorithms enable online transactions implemented in hardware and software. Unfortunately, more breaches have occurred in software-based multi-signature wallets \cite{boireau2018securing}. For cryptocurrency applications, for instance, owners access their credentials through a web-based authentication support or a hardware device that performs cryptographic operations. The software approach has suffered several setbacks regarding key exploitation by rogue impersonation. In 2016, the Bitfinex hack extorted private keys for multi-signature operations from different accounts to steal \$72 million worth of bitcoin \cite{baldwin2016bitcoin}. 
A very recent one is the Bybit multisignature compromise that stole Ethereum-based coins worth \$1.5 billion \cite{bybit2025hack}. As a result of unsecured software implementations, cryptocurrency has suffered a considerable drawback, and low adoption \cite{arapinis2019formal}. The hardware approach, such as TPM, is more tamper-resistant and aims to provide a secure environment for signing \cite{andrianagkaniatsou2017low}. TPMs are embedded with keys that never leave the device. The embedded keys are primarily used for private operations, such as generating signatures and performing other cryptographic operations. Unfortunately, existing multi-signature schemes impose requirements poorly suited to TPM hardware.

\subsection{Motivation}

TPM-based multi-signature is a Schnorr-type signature,
$Setup, KeyAgg, Sign, Verify$ executed within TPM 2.0 hardware, where private keys $x_i$ are non-exportable. State-of-the-art schemes like MuSig2~\cite{nick2020musig2} and FROST~\cite{komlo2020frost} have converged on a two-round structure: a commitment round where signers exchange nonce commitments, followed by a signing round. Both schemes support preprocessing, where nonce commitments can be generated and exchanged before the message is known, reducing the online phase to a single round. However, preprocessing in MuSig2 and FROST is inherently bounded. Each preprocessed session requires the signer to store a nonce pair. For $k$ preprocessed sessions, the signer must persist $k$ nonce pairs (64--128 bytes each). On TPM~2.0 hardware, where usable NVRAM is typically 6--16\,KB~\cite{svenda2024}, this limits preprocessing to roughly 50--125 sessions. Once exhausted, the signer must come online for an interactive commitment refill, re-introducing the synchronization and availability problems that preprocessing was meant to solve.

This creates a tension we call the \textbf{preprocessing gap}: bounded preprocessing helps with latency but does not solve the fundamental availability problem for resource-constrained devices with intermittent connectivity. A device that goes offline and returns to find its preprocessed nonces exhausted must participate in an interactive round before it can sign. ROAST~\cite{roast2022} addresses robustness in asynchronous networks but still relies on FROST's interactive commitment rounds. MuSig-DN~\cite{nick2020musigdn} enables deterministic nonces via NIZK proofs, but the proof exchange is itself interactive and computationally expensive.

\subsection{Our Contributions}

We present upTPM, a framework that closes the preprocessing gap through three contributions:

\begin{enumerate}
\item Unbounded deterministic preprocessing with constant storage: Each signer stores a single secret seed $\mathit{seed}_i$ (32 bytes) inside its TPM. From this seed, the signer deterministically derives an unlimited sequence of nonce scalars and publishes commitments. The nonce scalars never leave the TPM. We prove EU-CMA security under the DL assumption and PRF security in the random oracle model, with a one-time-use invariant enforced by TPM hardware state.

\item Asynchronous commitment refill: upTPM allows any signer to unilaterally generate and publish a batch of new commitments at any time, instead of signers being online simultaneously to run a new commitment round. We prove that this preserves EU-CMA security.

\item Crash-recoverable index management: We design a two-counter scheme implemented as TPM~2.0 monotonic Non-Volatile (NV) counters that enforce strictly sequential nonce consumption. We provide a complete analysis of coordinator crash recovery, showing that no security property is violated under state loss and that forward progress is guaranteed whenever the TPM is functional.

\end{enumerate}

We further generalize our framework to $(t,n)$-threshold signatures using Shamir secret sharing. We additionally show how  TPM attestation quotes can serve as a hardware-backed authenticity and state-binding mechanism for preprocessing commitments, strengthening resistance to host-software compromise


\subsection{Outline}

Section~\ref{sec:related} reviews related work. Section~\ref{sec:prelim} defines preliminaries. Section~\ref{sec:framework} presents the protocol with detailed index management and crash recovery analysis. 
Section~\ref{sec:threshold} gives the threshold extension and analyzes the coordinator trust model. Section~\ref{sec:security} provides the full security analysis. Section~\ref{sec:performance} evaluates performance. Section~\ref{sec:conclusion} presents the conclusion.

\section{\uppercase{Related Work}}
\label{sec:related}
Cryptographic signature schemes are an essential factor in controlling the access, authenticity, and integrity of users in any communication network \cite{schnorr1991efficient}. Early digital signature scheme of Schnorr \cite{schnorr1991efficient}, based on the difficulty of the discrete logarithm problem of large primes $p$ in $\mathbb{Z}_{p}$ provided authenticity and verification.

Schnorr multi-signatures have now evolved from the three-round Bellare--Neven scheme~\cite{bellare2006multi}. 
However, the works of Bellare--Neven \cite{bellare2006multi} provided a more practical approach by eliminating the key generation protocol of \cite{micali2001accountable}, where each signer in the group of signers generates a distinct challenge for partial signature generation. Successful implementation of the BN scheme gives rise to other variants, such as the Bagherzandi scheme \cite{bagherzandi2008multisignatures} that implemented the homomorphic function in the first round of the commitment stage when generating multi-signatures. Ma et al., \cite{ma2010efficient} proposed a two-round multi-signature scheme proven secure in the random oracle model with better signature sizes compared to \cite{bellare2006multi,bagherzandi2008multisignatures}. Maxwell et al., \cite{maxwell2019} proposed a variant of \cite{schnorr1991efficient} that supports key aggregation and is proven secure under the Discrete Logarithm assumption. The scheme proposed by \cite{maxwell2019} improves on the works of \cite{bellare2006multi,bagherzandi2008multisignatures,ma2010efficient} by supporting key aggregation and producing an efficient signature with the same signature size as \cite{schnorr1991efficient}. Nick et al., \cite{nick2020musig2} finally proposed a variant multi-signature scheme of \cite{maxwell2019}. Their proposed scheme was a novel two-round scheme that supports key aggregation and security based on the underlying multi-signatures of \cite{schnorr1991efficient}. MuSig2 supports preprocessing where the first round is executed before the message is known, but each preprocessed session consumes signer storage. A key insight of MuSig2 is the use of multiple nonce pairs per session to achieve concurrent security; our construction is compatible with this approach (each preprocessed index can store a tuple of commitments if needed). MuSig-DN~\cite{nick2020musigdn} enables deterministic nonces via NIZK proofs of correct derivation, but the proof exchange remains interactive.

For threshold signatures, FROST~\cite{komlo2020frost} provides two-round threshold Schnorr signatures with preprocessing support. ROAST~\cite{roast2022} wraps FROST to achieve robustness in asynchronous networks, but the underlying commitment exchange remains interactive. The IETF has standardized FROST as RFC~9591~\cite{rfc9591}. Boldyreva~\cite{boldyreva2003} proposed threshold BLS signatures, but these require pairing-friendly curves. Regarding TPM-based cryptography, Chen and Li~\cite{chen2013} proposed flexible TPM~2.0 signatures. Hedabou and Abdulsalam~\cite{hedabou2020} implemented BLS multi-signatures on TPM, but BLS requires pairing-friendly curves not natively supported by TPM~2.0. Svenda et al.~\cite{svenda2024} conducted a wide-scale study of TPM~2.0 security properties, revealing nonce-related vulnerabilities that underscore the importance of correct nonce handling. Chen et al.~\cite{chen2014} proposed cloud-based TPM collaboration.

All existing Schnorr multi-signature preprocessing schemes are \emph{bounded}: storage scales linearly with the number of preprocessed sessions. upTPM achieves \emph{unbounded} preprocessing with constant signer storage by keeping nonce-generating seeds inside the TPM and publishing only commitments. The coordinator is untrusted, matching the MuSig2~\cite{nick2020musig2}, FROST~\cite{komlo2020frost}, and ROAST \cite{roast2022} security model.

\section{\uppercase{Preliminaries}}
\label{sec:prelim}

\subsection{Notation}

Let $\mathbb{G}$ be a cyclic group of prime order $p$ with generator $g$. We write $\mathbb{Z}_p$ for integers modulo $p$. Let $H: \{0,1\}^* \to \mathbb{Z}_p$ be a hash function modeled as a random oracle, and $H_{\mathrm{agg}}: \{0,1\}^* \to \mathbb{Z}_p$ be an independent hash function for key aggregation. $\mathsf{PRF}: \{0,1\}^\lambda \times \{0,1\}^* \to \mathbb{Z}_p$ is a pseudorandom function family. $\|$ denotes concatenation. Table~\ref{tab:notation} summarizes the notation.

\begin{table}[ht]
\centering
\caption{Notation summary}
\label{tab:notation}
\small
\begin{tabularx}{\columnwidth}{lX}
\toprule
\textbf{Symbol} & \textbf{Description} \\
\midrule
$x_i, pk_i$ & Signer $i$'s private key and public key $g^{x_i}$ \\
$\mathit{seed}_i$ & Signer $i$'s secret nonce seed (inside TPM) \\
$w_i^{(j)}, R_i^{(j)}$ & $j$-th nonce scalar and commitment for signer $i$ \\
$\mathit{gen}_i$ & Generation counter: highest index ever created \\
$\mathit{next}_i$ & Consumption counter: next index to be consumed \\
$a_i$ & Aggregation coefficient $H_{\mathrm{agg}}(L \| pk_i)$ \\
$\mathit{apk}$ & Aggregated public key $\prod pk_i^{a_i}$ \\
$S$ & Signer set; $L = \mathsf{sort}(\{pk_j\}_{j \in S})$ \\
$B$ & Preprocessing batch size \\
\bottomrule
\end{tabularx}
\end{table}

\subsection{Formal Definitions}

\begin{definition}[Discrete Logarithm Assumption]
\label{def:dl}
For any PPT algorithm $\mathcal{A}$: $\Pr[x \gets \mathcal{A}(g, g^x) : x \in \mathbb{Z}_p] \leq \mathsf{negl}(\lambda)$ where $x \gets \mathbb{Z}_p$ uniformly.
\end{definition}

\begin{definition}[Pseudorandom Function]
\label{def:prf}
A keyed function family $\{\mathsf{PRF}(k, \cdot)\}_{k \in \{0,1\}^\lambda}$ is secure if for any PPT distinguisher $\mathcal{D}$:
\[
|\Pr_{k \gets \{0,1\}^\lambda}[\mathcal{D}^{\mathsf{PRF}(k,\cdot)} = 1] - \Pr_{f \gets \mathcal{F}}[\mathcal{D}^{f(\cdot)} = 1]| \leq \mathsf{negl}(\lambda)
\]
where $\mathcal{F}$ is the set of all functions from $\{0,1\}^*$ to $\mathbb{Z}_p$.
\end{definition}

\begin{definition}[EU-CMA Security for upTPM]
\label{def:eucma}
The upTPM multi-signature scheme is EU-CMA secure if for any PPT adversary $\mathcal{A}$ using the following experiment outputs $1$ with at most negligible probability:

\smallskip\noindent\underline{$\mathsf{Exp}^{\mathrm{eu\text{-}cma}}_\mathcal{A}(\lambda)$:}
\begin{enumerate}[label=\arabic*.,itemsep=0pt,leftmargin=2em]
\item Run $\mathsf{Setup}(1^\lambda, n)$. Give all public keys to $\mathcal{A}$.
\item $\mathcal{A}$ adaptively issues queries:
  \begin{itemize}[itemsep=0pt]
  \item $\mathcal{O}_{\mathrm{Corrupt}}(i)$: reveals $\mathit{seed}_i$ to $\mathcal{A}$ and gives $\mathcal{A}$ full control over signer $i$'s protocol messages. Does \emph{not} reveal $x_i$.
  \item $\mathcal{O}_{\mathrm{Preprocess}}(i)$: returns the next batch of commitments for honest signer $i$.
  \item $\mathcal{O}_{\mathrm{Sign}}(M, S)$: if $S$ contains at least one honest signer and unused commitments exist, returns a valid signature $\sigma$ on $M$ under signer set $S$. Records $M$ in set $\mathcal{Q}$.
  \end{itemize}
\item $\mathcal{A}$ outputs $(M^*, S^*, \sigma^*)$.
\item Output $1$ iff: $\mathsf{Verify}(M^*, \sigma^*, \{pk_i\}_{i \in S^*}) = 1$, and $M^* \notin \mathcal{Q}$, and $S^*$ contains at least one honest signer.
\end{enumerate}
\end{definition}

\begin{remark}
    The adversary observes all public data 
    The signing oracle is a black box that returns only the aggregate signature $\sigma = (R,s)$; the adversary does not get the honest signers' individual partial signature shares $s_i$ as a separate oracle output; they are sent through authenticated channels.
\end{remark}

\subsection{TPM 2.0 Primitives}

We use the following TPM~2.0 operations:
\begin{itemize}[leftmargin=*,itemsep=1pt]
\item $\mathsf{TPM2\_Create}$: Generates a non-exportable key pair $(x_i, pk_i)$.
\item $\mathsf{TPM2\_NV\_DefineSpace}$ / $\mathsf{NV\_Write}$: Allocates and writes NVRAM.
\item $\mathsf{TPM2\_NV\_Increment}$: Atomically increments a monotonic NV counter. The TPM specification~\cite{tpm_spec} guarantees the counter never decreases, even across power loss or reboots.
\item $\mathsf{TPM2\_Quote}$: Produces a signed attestation over specified data, bound to the TPM's attestation key (AK).
\item $\mathsf{TPM2\_HMAC}$: Computes HMAC inside the TPM using a loaded key, without exposing key material to host software.
\end{itemize}

\section{\uppercase{The upTPM Framework}}
\label{sec:framework}
Our proposed upTPM framework separates the protocol into three phases: \emph{setup}, \emph{preprocessing} (offline, asynchronous), and \emph{online signing} (single round). The coordinator is an untrusted server that stores public commitments and routes messages, as shown in Figure \ref{fig:arch}. 

\begin{figure}[ht]
\centering
\begin{tikzpicture}[
  node distance=0.4cm and 0.5cm,
  box/.style={rectangle, draw, minimum width=1.4cm, minimum height=0.6cm, font=\scriptsize, align=center},
  bigbox/.style={rectangle, draw, minimum width=2.2cm, minimum height=0.7cm, font=\scriptsize, align=center},
]
\node[box, fill=blue!10] (t1) {TPM$_1$\\$\mathit{seed}_1, x_1$};
\node[box, fill=blue!10, right=0.4cm of t1] (t2) {TPM$_2$\\$\mathit{seed}_2, x_2$};
\node[font=\scriptsize, right=0.3cm of t2] (dots) {$\cdots$};
\node[box, fill=blue!10, right=0.3cm of dots] (tn) {TPM$_n$\\$\mathit{seed}_n, x_n$};
\node[bigbox, fill=orange!10, below=0.8cm of t2] (coord) {Coordinator (untrusted)\\Stores: $\{R_i^{(j)}\}$, public only};
\node[bigbox, fill=yellow!15, above=0.6cm of t2] (bc) {Verifier};
\draw[->, thick, blue] (t1) -- node[left,font=\tiny] {$R_1^{(j)}$} (coord);
\draw[->, thick, blue] (t2) -- node[right,font=\tiny] {$R_2^{(j)}$} (coord);
\draw[->, thick, blue] (tn) -- node[right,font=\tiny] {$R_n^{(j)}$} (coord);
\draw[->, thick] (coord) -- +(0,-0.6) node[below,font=\tiny] {$\sigma = (R,s)$} ;
\end{tikzpicture}
\caption{upTPM architecture}
\label{fig:arch}
\end{figure}


A critical design point is that upTPM requires \emph{two} conceptually distinct counters per signer, both maintained inside the TPM:

\begin{itemize}[leftmargin=*,itemsep=2pt]
\item $\mathit{gen}_i$: The \textbf{generation counter}. Highest nonce index ever created by TPM$_i$. Advances during preprocessing. Invariant: commitments have been published for indices $1, \ldots, \mathit{gen}_i$.
\item $\mathit{next}_i$: The \textbf{consumption counter}. Smallest index not yet used for signing. Advances during signing. Invariant: indices $1, \ldots, \mathit{next}_i - 1$ have been consumed and will never be reused.
\end{itemize}

Both counters are monotonically non-decreasing and satisfy $1 \leq \mathit{next}_i \leq \mathit{gen}_i + 1$ at all times. The set of available (preprocessed but unconsumed) indices for signer $i$ is $\{\mathit{next}_i, \ldots, \mathit{gen}_i\}$. We adopt strictly sequential consumption, where the coordinator must assign indices in order, and the TPM only accepts a signing request for index $j = \mathit{next}_i$, after which it increments $\mathit{next}_i$. Both $\mathit{gen}_i$ and $\mathit{next}_i$ are implemented as TPM~2.0 NV counters with $\mathsf{TPMA\_NV\_COUNTER}$ attribute, which the TPM hardware guarantees to be atomic and monotonic even across power loss.

\begin{remark}[Why strictly sequential consumption is sufficient]
\label{rem:sequential}
Sequential consumption allows the one-time-use invariant to be enforced by a single monotonic counter, which is the simplest mechanism available in TPM~2.0 hardware. An alternative design allowing arbitrary-order consumption would require the TPM to maintain a bitmap or hash set of used indices, which consumes more NVRAM.
\end{remark}

\subsection*{Setup}
\label{sec:setup}

\begin{enumerate}[leftmargin=*,itemsep=1pt]
\item \textbf{Input:} Security parameter $1^\lambda$, signer count $n$.
\item For each signer $i = 1, \ldots, n$:
  \begin{enumerate}[label=(\alph*),itemsep=0pt]
    \item Generate key pair inside TPM: $(x_i, pk_i) \gets \mathsf{TPM2\_Create}()$ with $pk_i = g^{x_i}$. The private key $x_i$ is non-exportable.
    \item Generate nonce seed inside TPM: $\mathit{seed}_i \gets \{0,1\}^\lambda$ via $\mathsf{TPM2\_GetRandom}$, stored in NVRAM, bound to TPM's storage hierarchy, never exposed to host software.
    \item Initialize counters: $\mathit{gen}_i \gets 0$, $\mathit{next}_i \gets 1$ (both NV counters).
  \end{enumerate}
\item All signers broadcast public keys.
\item \textbf{Output:} Public parameters $\mathit{pp} = (g, p, \mathbb{G}, n)$; each TPM holds $(\mathit{seed}_i, x_i, \mathit{gen}_i, \mathit{next}_i)$ internally.
\end{enumerate}

\noindent\textbf{Key aggregation} ($\mathsf{KeyAgg}$, public, no secrets needed). For signer set $S$ with sorted key list $L = \mathsf{sort}(\{pk_j\}_{j \in S})$:
\begin{align}
a_i &= H_{\mathrm{agg}}(L \| pk_i) \quad \forall\, i \in S \label{eq:coeff} \\
\mathit{apk} &= \textstyle\prod_{i \in S} pk_i^{a_i} \label{eq:apk}
\end{align}
This is identical to MuSig2's key aggregation~\cite{nick2020musig2}, using a public hash function. Rogue-key resistance follows from the random oracle model (Theorem~\ref{thm:rogue}).

\subsection*{Preprocessing Phase (Offline, Asynchronous)}
\label{sec:preprocess}

Each signer $i$ independently performs preprocessing whenever it is online:

\begin{enumerate}[leftmargin=*,itemsep=1pt]
\item \textbf{Inside the TPM,} for each index $j = \mathit{gen}_i + 1, \ldots, \mathit{gen}_i + B$:
  \begin{enumerate}[label=(\alph*),itemsep=0pt]
    \item Compute nonce scalar with domain separation:
    \begin{equation}\label{eq:nonce}
    w_i^{(j)} = \mathsf{PRF}(\mathit{seed}_i,\; \texttt{"upTPM-nonce"} \| pk_i \| j)
    \end{equation}
    computed inside the TPM via $\mathsf{TPM2\_HMAC}$. The scalar $w_i^{(j)}$ is \textbf{not exported} to host software.
    \item Compute commitment: $R_i^{(j)} = g^{w_i^{(j)}}$. This group element \textbf{is} exported.
    \item (Optional) Produce attestation: $\mathit{att}_i^{(j)} = \mathsf{TPM2\_Quote}(\mathit{AK}_i,\; H(R_i^{(j)} \| pk_i \| j))$.
  \end{enumerate}
\item Advance generation counter: $\mathit{gen}_i \gets \mathit{gen}_i + B$ via $\mathsf{TPM2\_NV\_Increment}$.
\item \textbf{Publish to coordinator:} Send $\{(R_i^{(j)}, j, \mathit{att}_i^{(j)})\}_{j}$.
\end{enumerate}

\noindent\textbf{Critical security property:} $w_i^{(j)}$ is computed and consumed entirely inside the TPM. The host and coordinator receive only $R_i^{(j)} = g^{w_i^{(j)}}$.

\begin{remark}[TPM command mapping] 
The preprocessing computation requires evaluating $\mathsf{PRF}(\mathit{seed}_i, \cdot)$ and computing $g^{w_i^{(j)}}$ without exporting $w_i^{(j)}$, such that the nonce derivation and commitment generation are bound to TPM-internal state. In commodity TPM~2.0, this binding is not exposed as a single native command for arbitrary PRF-derived scalar multiplication. The \textsf{TPM2\_Commit} and \textsf{TPM2\_Sign} commands provide a partial realization, 
though without caller-controlled seed-based derivation. Accordingly, our protocol treats this as an instantiation layer: the security analysis applies to the protocol abstraction. The deployment requires either the native TPM commit workflow, vendor-specific protected execution support, or a TEE-assisted realization that preserves non-exportability of the nonce scalar.
\end{remark}

\noindent\textbf{Storage:} Each TPM stores $\mathit{seed}_i$ (32\,B) + $x_i$ (32\,B) + $\mathit{gen}_i$ (8\,B) + $\mathit{next}_i$ (8\,B) = \textbf{80 bytes total}, independent of the number of preprocessed commitments.

\subsection*{Online Signing (Single Round)}
\label{sec:signing}

Given message $M$ and signer set $S$:

\begin{enumerate}[leftmargin=*,itemsep=1pt]
\item \textbf{Coordinator:}
  \begin{enumerate}[label=(\alph*),itemsep=0pt]
    \item For each $i \in S$, select $j_i = \mathit{next}_i$ (coordinator's record of signer $i$'s next unconsumed index) and retrieve $R_i^{(j_i)}$ from storage. If no commitment is available, abort and request preprocessing.
    \item (If attested) Verify $\mathit{att}_i^{(j_i)}$ for each signer.
    \item Compute $R = \prod_{i \in S} R_i^{(j_i)}$.
    \item Send $(M, R, \{(R_k^{(j_k)}, j_k)\}_{k \in S})$ to each signer.
  \end{enumerate}

\item \textbf{Each signer $i \in S$ (inside TPM):}
  \begin{enumerate}[label=(\alph*),itemsep=0pt]
    \item \textbf{Check index:} Verify $j_i = \mathit{next}_i$. If $j_i \neq \mathit{next}_i$, \textbf{abort}.
    \item Recompute $w_i^{(j_i)} = \mathsf{PRF}(\mathit{seed}_i, \texttt{"upTPM-nonce"} \| pk_i \| j_i)$ inside TPM.
    \item Recompute $R_i^{(j_i)} = g^{w_i^{(j_i)}}$ and verify it matches the received value. If mismatch, \textbf{abort}.
    \item Verify $R = \prod_{k \in S} R_k^{(j_k)}$. If mismatch, \textbf{abort}.
    \item Compute $(\mathit{apk}, \{a_k\}) \gets \mathsf{KeyAgg}(\{pk_k\}_{k \in S})$.
    \item Compute $c = H(R \| \mathit{apk} \| M)$.
    \item Compute $s_i = w_i^{(j_i)} + c \cdot a_i \cdot x_i \bmod p$.
    \item \textbf{Advance consumption counter:} $\mathit{next}_i \gets \mathit{next}_i + 1$ via $\mathsf{TPM2\_NV\_Increment}$.
    \item Output $s_i$ to host software.
  \end{enumerate}

\item \textbf{Aggregation:} Collect $\{s_i\}_{i \in S}$; optionally verify $g^{s_i} \stackrel{?}{=} R_i^{(j_i)} \cdot pk_i^{a_i \cdot c}$; compute $s = \sum_{i \in S} s_i \bmod p$.
\item \textbf{Output:} $\sigma = (R, s)$.
\end{enumerate}

\noindent\textbf{One-time-use invariant.} After step 2(h), the TPM has irrevocably incremented $\mathit{next}_i$, so index $j_i$ can never be accepted again.

\subsection*{Verification}

Given $(M, \sigma = (R, s), \{pk_i\}_{i \in S})$: compute $(\mathit{apk}, \{a_i\}) \gets \mathsf{KeyAgg}$, compute $c = H(R \| \mathit{apk} \| M)$, accept iff $g^s = R \cdot \mathit{apk}^c$. 

\subsection{Coordinator State and Crash Recovery}
\label{sec:recovery}

The coordinator maintains a database of published commitments $\{(i, j, R_i^{(j)}, \mathit{att}_i^{(j)}, \mathit{used})\}$ and its local record of each signer's next unconsumed index. This is \emph{public, non-secret state}: losing it does not compromise any private key or nonce scalar. We consider the following three scenarios:

\begin{itemize}
    \item Coordinator crash: If the coordinator loses state, commitments can be re-published by signers from their seeds (the TPM can regenerate $R_i^{(j)}$ for any $j \leq \mathit{gen}_i$). The coordinator recovers by querying each signer's TPM for its current $\mathit{gen}_i$ and $\mathit{next}_i$ values and requesting re-publication of commitments in the range $[\mathit{next}_i, \mathit{gen}_i]$.
    
    \item Coordinator--TPM desynchronization: If the coordinator's local record of $\mathit{next}_i$ falls behind the TPM's actual $\mathit{next}_i$ (e.g., because the coordinator recorded a signing session as incomplete but the TPM already incremented), the coordinator may attempt to use an already-consumed index. The TPM will reject this with an abort (step~2a), at which point the coordinator queries the TPM for its current $\mathit{next}_i$ and resynchronizes. No security property is violated; at worst, one signing attempt fails and is retried.
    
    \item Forward progress guarantee: As long as the TPM is functional and has $\mathit{next}_i \leq \mathit{gen}_i$, the protocol can always make progress--the coordinator queries $\mathit{next}_i$, selects the corresponding commitment, and proceeds.
\end{itemize}

\section{\uppercase{upTPM Threshold Signature Extension}}
\label{sec:threshold}

We extend upTPM to $(t,n)$-threshold signatures using Shamir \cite{shamir1979share} secret sharing with Feldman VSS \cite{feldman1987practical}. Our proposed threshold extension demonstrates that the unbounded preprocessing approach generalizes. 
The protocol is extended as follows:

\subsection*{Distributed Key Generation}

\begin{enumerate}[leftmargin=*,itemsep=1pt]
\item Each party $i \in \{1,\ldots,n\}$ chooses a random degree-$(t{-}1)$ polynomial $\phi_i(X) = \phi_{i,0} + \phi_{i,1}X + \cdots + \phi_{i,t-1}X^{t-1}$ over $\mathbb{Z}_p$.
\item Party $i$ computes Feldman commitments $C_{i,k} = g^{\phi_{i,k}}$ for $k = 0, \ldots, t{-}1$ and broadcasts these to all parties.
\item For each $j \neq i$, party $i$ sends the share $\phi_i(j)$ to party $j$ over a secure channel.
\item Party $j$ verifies each received share: $g^{\phi_i(j)} \stackrel{?}{=} \prod_{k=0}^{t-1} C_{i,k}^{j^k}$. If verification fails, party $j$ broadcasts a complaint and the DKG restarts.
\item After successful verification, party $j$ computes its combined share $x_j = \sum_{i=1}^n \phi_i(j)$ and stores $x_j$ inside TPM$_j$ (non-exportable). The public verification share is $Y_j = g^{x_j} = \prod_{i=1}^n \prod_{k=0}^{t-1} C_{i,k}^{j^k}$ (publicly computable).
\item Master public key: $\mathit{mpk} = \prod_{i=1}^n C_{i,0} = g^{\sum_i \phi_{i,0}} = g^x$ where $x = \sum_i \phi_{i,0}$.
\item Each TPM$_j$ generates its nonce seed $\mathit{seed}_j$ and initializes $\mathit{gen}_j = 0$, $\mathit{next}_j = 1$, exactly as in the multi-signature setup.
\end{enumerate}

\subsection*{Preprocessing}

Identical to Section~\ref{sec:preprocess}: each TPM$_j$ independently generates commitments $R_j^{(k)} = g^{w_j^{(k)}}$ where $w_j^{(k)} = \mathsf{PRF}(\mathit{seed}_j, \texttt{"upTPM-nonce"} \| pk_j \| k)$, and publishes them to the coordinator.

\subsection*{Threshold Signing}

Given message $M$ and threshold set $T$ with $|T| \geq t+1$:

\begin{enumerate}[leftmargin=*,itemsep=1pt]
\item \textbf{Coordinator:} For each $i \in T$, retrieve next unconsumed commitment $R_i^{(j_i)}$. Compute Lagrange coefficients:
\begin{equation}\label{eq:lagrange}
\lambda_i = \prod_{k \in T, k \neq i} \frac{-k}{i-k} \bmod p
\end{equation}
Compute weighted commitment $R = \prod_{i \in T} (R_i^{(j_i)})^{\lambda_i}$. Send $(M, R, \{(R_k^{(j_k)}, j_k, \lambda_k)\}_{k \in T})$ to all signers in $T$.

\item \textbf{Each signer $i \in T$ (inside TPM):}
  \begin{enumerate}[label=(\alph*),itemsep=0pt]
    \item Verify $j_i = \mathit{next}_i$.
    \item Recompute $w_i^{(j_i)}$ and $R_i^{(j_i)}$; verify consistency.
    \item Verify Lagrange coefficients $\lambda_i$ by recomputing from $T$.
    \item Verify $R = \prod_{k \in T} (R_k^{(j_k)})^{\lambda_k}$.
    \item Compute $c = H(R \| \mathit{mpk} \| M)$.
    \item Compute $s_i = \lambda_i \cdot w_i^{(j_i)} + c \cdot \lambda_i \cdot x_i \bmod p$.
    \item Advance $\mathit{next}_i$.
    \item Output $s_i$.
  \end{enumerate}

\item \textbf{Reconstruction:} $s = \sum_{i \in T} s_i \bmod p$.
\item \textbf{Output:} $\sigma = (R, s)$.
\end{enumerate}

\subsection*{Threshold Verification}

Given $(M, \sigma = (R,s), \mathit{mpk})$: accept iff $g^s = R \cdot \mathit{mpk}^c$ where $c = H(R \| \mathit{mpk} \| M)$.

\subsection*{Coordinator Trust Model}
\label{sec:trust}

The coordinator is untrusted for unforgeability. It stores public commitments $\{R_i^{(j)}\}$, public keys, attestation quotes, and usage metadata.The coordinator receives all partial signatures $\{s_i\}_{i \in S}$ (from both honest and corrupted signers) in order to compute the aggregate $s = \sum s_i$. However, in the EU-CMA security model (Definition~3), the 
adversary is given only the aggregate signature $\sigma = (R,s)$  as oracle output; individual honest signers' shares are not 
separately exposed. This model's deployments were honest signers 
send their shares over authenticated private channels to the 
aggregation point. We consider the following scenarios:

\begin{itemize}
    \item Cannot forge: $w_i^{(j)} = \log_g(R_i^{(j)})$ requires solving DL. Cannot extract $x_i$ from $s_i = w_i^{(j)} + c \cdot a_i \cdot x_i$ without $w_i^{(j)}$, which it does not know.
    
    \item Cannot cause nonce reuse: TPM's $\mathit{next}_i$ advances monotonically.
    
    \item Can deny service: Mitigated by replicating coordinator state (which is entirely public).
    
    \item Can substitute commitments: Detected by signer verification (step~2c--2d).
    
    \item State loss: Coordinator state is non-secret. Recovery is described in Section~\ref{sec:recovery}.
\end{itemize}

\section{\uppercase{Security Analysis}}
\label{sec:security}
We emphasize that upTPM does not use message-deterministic nonce 
generation (e.g., RFC~6979-style derandomization). Instead, each 
signer derives a pseudorandom per-index nonce scalar from a 
TPM-sealed seed via $w_i^{(j)} = \mathsf{PRF}(\mathit{seed}_i, 
\texttt{"upTPM-nonce"} \| pk_i \| j)$, and the TPM enforces strict 
one-time use via a monotone consumption counter. This prevents the 
classic key-recovery attack arising from nonce reuse across distinct 
challenges. Our construction is consistent with BIP~340~\cite{bip340} and RFC~9591~\cite{rfc9591} warnings that naive deterministic nonce derivation is unsafe in multiparty Schnorr protocols.

\subsection{Correctness}

\begin{theorem}[Multi-signature Correctness]
\label{thm:correct}
If all parties follow the protocol honestly, verification accepts.
\end{theorem}

\begin{proof}
By construction: $s = \sum_{i \in S} s_i = \sum_{i \in S} (w_i^{(j_i)} + c \cdot a_i \cdot x_i)$. Separating: $s = \sum_{i \in S} w_i^{(j_i)} + c \cdot \sum_{i \in S} a_i x_i$.

The aggregated commitment satisfies $R = \prod_{i \in S} g^{w_i^{(j_i)}} = g^{\sum_{i \in S} w_i^{(j_i)}}$, so $\sum_{i \in S} w_i^{(j_i)} = \log_g(R)$.

The aggregated public key satisfies $\mathit{apk} = \prod_{i \in S} g^{a_i x_i} = g^{\sum_{i \in S} a_i x_i}$, so $\sum_{i \in S} a_i x_i = \log_g(\mathit{apk})$.

Therefore $g^s = g^{\log_g(R) + c \cdot \log_g(\mathit{apk})} = R \cdot \mathit{apk}^c$. \qedhere
\end{proof}

\begin{theorem}[Threshold-signature Correctness]
For honest execution with $|T| \geq t+1$, verification accepts.
\end{theorem}

\begin{proof}
$s = \sum_{i \in T} (\lambda_i w_i^{(j_i)} + c \lambda_i x_i) = \sum_{i \in T} \lambda_i w_i^{(j_i)} + c \sum_{i \in T} \lambda_i x_i$.

$R = \prod_{i \in T} (g^{w_i^{(j_i)}})^{\lambda_i} = g^{\sum \lambda_i w_i^{(j_i)}}$, so $\log_g(R) = \sum \lambda_i w_i^{(j_i)}$.

By Lagrange interpolation: $\sum_{i \in T} \lambda_i x_i = x$ and $\mathit{mpk} = g^x$.

Therefore $g^s = g^{\log_g(R) + c \cdot x} = R \cdot \mathit{mpk}^c$.
\end{proof}

\subsection{EU-CMA Security}

We prove EU-CMA security of upTPM in the random oracle model under the discrete logarithm (DL) assumption and PRF security, assuming the one-time-use invariant enforced by the TPM consumption counter $\text{next}_i$. We consider the natural aggregate-signature view of the protocol: the adversary sees the public commitments $R_i^{(j)}$, public keys, coordinator messages,
and final aggregate signatures, may adaptively corrupt signers to learn their seeds, and may control the coordinator. For honest signers, however, nonce scalars remain inside the TPM, and honest partial signatures are not exposed as an oracle output before aggregation.

\begin{theorem}[Multi-signature EU-CMA Security]
\label{thm:eucma}
upTPM is EU-CMA secure in the random oracle model under the DL assumption and PRF security, assuming the one-time-use invariant enforced by TPM hardware counters. Specifically, for any PPT adversary $\mathcal{A}$ in the experiment of Definition~3, there exist PPT algorithms $\mathcal{B}$
and $\mathcal{C}$ such that
\[
\Pr[\mathrm{Exp}^{\mathrm{eu\text{-}cma}}_{\mathcal{A}} = 1]
\le
n \cdot \left(
\mathrm{Adv}^{\mathrm{dl}}_{\mathcal{B}}(\lambda)
+ \frac{q_S (q_H + q_S)}{p}
\right)
+ n \cdot \mathrm{Adv}^{\mathrm{prf}}_{\mathcal{C}}(\lambda)
+ \mathrm{negl}(\lambda),
\]
where $q_H$ is the number of random-oracle queries, $q_S$ is the number of signing queries,
$n$ is the number of honest signers, and $p$ is the order of the group.
\end{theorem}

\begin{proof}
We proceed by a sequence of games.

\textbf{Game $G_0$.}
This is the real EU-CMA experiment of Definition~3. Let $\epsilon_0$ denote the adversary’s
success probability.

\textbf{Game $G_1$.}
For each honest signer $i$, replace $\mathrm{PRF}(\text{seed}_i, \cdot)$ with an independent uniformly random
function $f_i(\cdot) : \{0,1\}^* \to \mathbb{Z}_p$. By a standard hybrid argument over the honest signers,
\[
|\Pr[G_0] - \Pr[G_1]| \le n \cdot \mathrm{Adv}^{\mathrm{prf}}_{\mathcal{C}}(\lambda).
\]

In Game $G_1$, every honest nonce scalar $w_i^{(j)}$ is independent and uniform in $\mathbb{Z}_p$,
hence every commitment $R_i^{(j)} = g^{w_i^{(j)}}$ is a uniform group element. Because the TPM enforces
one-time use through the strictly monotone counter $\text{next}_i$, each honest $w_i^{(j)}$ appears in at most
one signing equation.

\textbf{Game $G_2$.}
Replace $H$ and $H_{\mathrm{agg}}$ with programmable random oracles. This is only a syntactic
reformulation of the random oracle model, so the adversary’s success probability is unchanged.

\textbf{Reduction to DL.}
We now construct a DL solver $\mathcal{B}$ using a successful forger in Game $G_2$. The solver receives
a challenge $(g, Y = g^y)$ and must recover $y$.

$\mathcal{B}$ guesses one honest signer $j^*$ that appears in the eventual forgery. This guess is correct
with probability at least $1/n$. For every other honest signer $i \neq j^*$, $\mathcal{B}$ chooses
$x_i \leftarrow \mathbb{Z}_p$ and sets $pk_i = g^{x_i}$.

For the target signer $j^*$, $\mathcal{B}$ embeds the DL challenge into the public key as follows.
When $H_{\mathrm{agg}}$ is first queried on a key list $L$ containing $pk_{j^*}$, $\mathcal{B}$ programs
\[
a_{j^*} = H_{\mathrm{agg}}(L \parallel pk_{j^*})
\]
as a fresh nonzero value and sets
\[
pk_{j^*} = Y^{a_{j^*}^{-1}}.
\]

Then
\[
pk_{j^*}^{a_{j^*}} = Y = g^y,
\]
so the target weighted secret contribution satisfies $a_{j^*} x_{j^*} = y$.

In Game $G_1$, preprocessing commitments for honest signers are uniformly random group elements.
Therefore $\mathcal{B}$ may simulate preprocessing for the target signer by choosing independent
uniform exponents $\tilde{w}_{j^*}^{(j)} \leftarrow \mathbb{Z}_p$ and publishing
\[
\tilde{R}_{j^*}^{(j)} = g^{\tilde{w}_{j^*}^{(j)}}.
\]

This is identically distributed to the real game after the PRF replacement. All other honest signers
are simulated honestly.

Signing queries are answered consistently using fresh preprocessed commitments and random-oracle
programming. The only bad event is that the adversary queries
$H(R \parallel apk \parallel M)$ before the simulator programs that point for a signing query.
Since each signing query uses a fresh aggregate commitment $R$ by one-time use, the probability of
this bad event over all signing queries is at most
\[
\frac{q_S (q_H + q_S)}{p}.
\]

Conditioned on no such collision, the adversary’s view is identically distributed to Game $G_2$.

Suppose now that the adversary outputs a valid forgery $(M^*, S^*, \sigma^* = (R^*, s^*))$
such that $M^*$ was not previously submitted to $O_{\mathrm{Sign}}$ and $S^*$ contains the target
signer $j^*$. Let
\[
c^* = H(R^* \parallel apk^* \parallel M^*), \quad
apk^* = \prod_{i \in S^*} pk_i^{a_i}.
\]

By validity of the forgery,
\[
g^{s^*} = R^* \cdot (apk^*)^{c^*}.
\]

By the general forking lemma, rewinding on the random-oracle query
$H(R^* \parallel apk^* \parallel M^*)$ yields, with non-negligible probability,
a second valid forgery $(R^*, s')$ on the same input but with challenge $c' \neq c^*$:
\[
g^{s'} = R^* \cdot (apk^*)^{c'}.
\]

Dividing the two equations gives
\[
g^{s^* - s'} = (apk^*)^{c^* - c'},
\]
and hence
\[
\log_g(apk^*) = (s^* - s')(c^* - c')^{-1} \pmod p.
\]

Since
\[
\log_g(apk^*) = \sum_{i \in S^*} a_i x_i
= a_{j^*} x_{j^*} + \sum_{i \in S^* \setminus \{j^*\}} a_i x_i,
\]
and $\mathcal{B}$ knows $x_i$ for every honest signer $i \neq j^*$ as well as all public contributions
of corrupted signers, it can recover
\[
a_{j^*} x_{j^*}
= \log_g(apk^*) - \sum_{i \in S^* \setminus \{j^*\}} a_i x_i.
\]

By construction, $a_{j^*} x_{j^*} = y$, so $\mathcal{B}$ solves the DL challenge.

Combining the loss from the target-signer guess, the PRF hybrid bound, and the random-oracle
collision term yields the claimed inequality.
\end{proof}

\begin{theorem}[Threshold EU-CMA Security]
\label{thm:thresheucma}
Under the DL assumption, PRF security, the one-time-use invariant, and the ROM, the $(t,n)$-threshold scheme is EU-CMA secure against any adversary corrupting up to $t$ signers, with an untrusted coordinator.
\end{theorem}

\begin{proof}
The structure follows Theorem~\ref{thm:eucma}. $\mathcal{B}$ receives $(g, Y = g^y)$ and sets $\mathit{mpk} = Y$.

\emph{DKG simulation.} For corrupted parties $C$ with $|C| \leq t$, $\mathcal{B}$ chooses shares $\{x_i\}_{i \in C}$ and provides them to $\mathcal{A}$. For honest parties, shares are implicitly defined by the constraint that interpolation over any $t+1$ shares yields $x = y$. Since $|C| \leq t$ and the polynomial has degree $t-1$, these $t$ shares leave one degree of freedom; honest shares are statistically determined but unknown to $\mathcal{A}$ (standard argument).

$\mathcal{B}$ simulates Feldman commitments for corrupted shares honestly. For honest shares, $\mathcal{B}$ computes $Y_i = g^{x_i}$ from the Feldman commitments (publicly computable) but does not know $x_i$.

\emph{PRF replacement.} As in Theorem~\ref{thm:eucma}, replace honest PRFs with random functions.

\emph{Signing simulation.} For signing queries on threshold set $T$ containing at least one honest signer: the argument follows Theorem~\ref{thm:eucma} with $\mathit{apk}$ replaced by $\mathit{mpk} = Y$. The aggregate signature is simulated using random oracle programming.

\emph{Forgery extraction.} The forking lemma yields $\log_g(\mathit{mpk}) = y$.
\end{proof}

\begin{remark}[Threshold Robustness]
\label{thres_robust}
If up to $t$ signers submit incorrect shares, the remaining honest signers can still produce a valid signature (assuming $\geq t+1$ honest signers are available). Since shares are verifiable: $g^{s_i} \stackrel{?}{=} (R_i^{(j_i)})^{\lambda_i} \cdot Y_i^{c \lambda_i}$ where $Y_i = g^{x_i}$ from Feldman commitments. Invalid shares are excluded, and reconstruction proceeds with any $t+1$ verified shares.    
\end{remark}

\subsection{Asynchronous Refill Security}
\label{sec:async}

\begin{theorem}[Asynchronous Refill]
\label{thm:async}
EU-CMA security (Theorem~\ref{thm:eucma}) holds when signers independently extend their commitment pools at arbitrary times.
\end{theorem}

\begin{proof}
In Game~1 of Theorem~\ref{thm:eucma}, each honest signer's nonces are independently random. The distribution of $w_i^{(j)}$ depends only on the random function $f_i$ and the index $j$, not on when $R_i^{(j)}$ is published. An adversary observing publication times learns only timing metadata.

The reduction in Theorem~\ref{thm:eucma} does not depend on publication timing: $\mathcal{B}$ generates commitments for honest signers as uniform random group elements regardless of when $\mathcal{A}$ requests them. The one-time-use invariant is enforced by $\mathit{next}_i$, which advances monotonically regardless of preprocessing schedule.
\end{proof}

\subsection{TPM-Attested Commitments and Rogue Key Resistance}
\label{sec:attested}

TPM attestation provides a hardware-backed authenticity and state-binding mechanism for preprocessing commitments.

\begin{theorem}[Attestation Guarantees]
\label{thm:attested}
If the TPM hardware of signer $i$ is uncompromised but its host software is fully adversary-controlled, then:
\begin{enumerate}[label=(\roman*),itemsep=0pt]
\item The adversary cannot substitute a commitment $R_i' \neq R_i^{(j)}$ that passes attestation verification.
\item The adversary cannot cause the TPM to produce a partial signature for an index $j < \mathit{next}_i$.
\item EU-CMA security (Theorem~\ref{thm:eucma}) is maintained: the information available to a compromised host 
falls within the adversary's view in the security model of Definition~3. 
\end{enumerate}
\end{theorem}

\begin{proof}
(i) The attestation quote $\mathit{att}_i^{(j)} = \mathsf{TPM2\_Quote}(\mathit{AK}_i, H(R_i^{(j)} \| pk_i \| j))$ is a signature under $\mathit{AK}_i$. Forging requires compromising $\mathit{AK}_i$ (TPM hardware).

(ii) The TPM checks $j = \mathit{next}_i$ before signing. The adversary can submit arbitrary requests, but the TPM rejects any $j \neq \mathit{next}_i$.

(iii) A compromised host observes $R_i^{(j)}$ and $s_i$ and controls message selection. This is exactly the information model in Theorem~\ref{thm:eucma}.
\end{proof}

\begin{remark}[Limitations of attestation]
Attestation does not by itself prove: that the TPM's RNG produced high-quality randomness for $\mathit{seed}_i$; that no physical side-channel attack has been mounted~\cite{svenda2024}; or that the host cannot trigger signing in a context the user did not intend (UI-level attacks). We present attestation as an additional defense layer that binds commitments to TPM identity and state, not as a complete guarantee against all compromise forms.
\end{remark}


\begin{theorem}[Rogue Key Resistance]
\label{thm:rogue}
The key aggregation $a_i = H_{\mathrm{agg}}(L \| pk_i)$ provides rogue-key resistance in the random oracle model.
\end{theorem}

\begin{proof}
This is identical to MuSig2~\cite{nick2020musig2,maxwell2019}. Aggregation coefficients are random oracle outputs, unpredictable when the adversary chooses rogue keys.
\end{proof}

\section{\uppercase{Performance Evaluation}}
\label{sec:performance}

We evaluate upTPM through a simulation-based benchmark suite implemented in Python~3.11 on Ubuntu~24.04 (AMD Ryzen 7, 32\,GB RAM) using the 
\texttt{ecdsa} library (v0.18) for secp256k1 group operations and comparing against MuSig2~\cite{nick2020musig2}, FROST~\cite{komlo2020frost} under identical group operations for fairness. MuSig2 is modeled with one nonce pair per preprocessing session (the single-nonce variant); the concurrent-security variant using multiple nonce pairs would increase MuSig2's per-session storage proportionally but does not affect the qualitative comparison. We also report results under simulated TPM hardware latencies calibrated from the TPMScan study~\cite{svenda2024} and the TPM~2.0 specification~\cite{tpm_spec}.


\subsection*{Storage Requirements}
\label{sec:eval_storage}

The primary quantitative advantage of upTPM is constant per-signer TPM storage regardless of the number of preprocessed sessions, as shown in Figure \ref{fig:storage}

\begin{figure}[ht]
\centering
\begin{tikzpicture}
\begin{semilogyaxis}[
  width=\columnwidth, height=4.5cm,
  xlabel={Preprocessed sessions ($k$)},
  ylabel={Per-signer storage (bytes)},
  legend style={at={(0.03,0.97)},anchor=north west,font=\scriptsize},
  grid=major, grid style={gray!30},
  xmin=0, xmax=550, ymin=50, ymax=50000,
]
\addplot[blue,thick,mark=o,mark size=1.5] coordinates {
  (10,80)(50,80)(100,80)(127,80)(200,80)(500,80)
};
\addlegendentry{upTPM (constant)}
\addplot[red,thick,mark=square,mark size=1.5] coordinates {
  (10,672)(50,3232)(100,6432)(127,8160)(200,12832)(500,32032)
};
\addlegendentry{MuSig2/FROST (linear)}
\addplot[gray,dashed,thick,domain=0:550] {8192};
\addlegendentry{TPM NVRAM limit (8\,KB)}
\end{semilogyaxis}
\end{tikzpicture}
\caption{Per-signer TPM NVRAM usage. upTPM stores 80\,B regardless of $k$. MuSig2/FROST exceeds the typical 8\,KB TPM NVRAM limit at $k = 128$.}
\label{fig:storage}
\end{figure}

upTPM stores $\mathit{seed}_i$ (32\,B) + $x_i$ (32\,B) + $\mathit{gen}_i$ (8\,B) + $\mathit{next}_i$ (8\,B) = 80\,B. MuSig2/FROST stores $x_i$ (32\,B) + $k$ nonce pairs (64\,B each) = $32 + 64k$\,B. With typical TPM NVRAM of 8\,KB, MuSig2 supports at most $\lfloor(8192-32)/64\rfloor = 127$ sessions. At $k = 127$, MuSig2 uses $102\times$ more storage than upTPM. At $k = 500$, the ratio grows to $400\times$. Practically, at one signing per day, MuSig2 exhausts its preprocessing in roughly four months, forcing all signers online simultaneously for interactive refill. upTPM faces no such deadline.

\subsection*{Online Signing Latency}
\label{sec:eval_signing}

upTPM's online signing involves more computation per signer than MuSig2 because each signer must \emph{recompute} the nonce scalar from the seed and \emph{verify} its commitment (Figre \ref{fig:signing_compute}). In MuSig2 with preprocessing, the nonce is already in memory.


\begin{figure}[ht]
\centering
\begin{tikzpicture}
\begin{axis}[
  width=\columnwidth, height=4.5cm,
  xlabel={Number of signers ($n$)},
  ylabel={Online signing latency (ms)},
  legend style={at={(0.03,0.97)},anchor=north west,font=\scriptsize},
  grid=major, grid style={gray!30},
  xmin=0, xmax=55,
]
\addplot[blue,thick,mark=o,mark size=1.5] coordinates {
  (2,4.09)(3,7.12)(5,12.32)(10,21.64)(20,44.98)(50,112.24)
};
\addlegendentry{upTPM}
\addplot[red,thick,mark=square,mark size=1.5] coordinates {
  (2,0.08)(3,0.10)(5,0.13)(10,0.17)(20,0.28)(50,0.62)
};
\addlegendentry{MuSig2 (preprocessed)}
\end{axis}
\end{tikzpicture}
\caption{Online signing latency (computation only, software). upTPM's cost is dominated by the commitment recomputation exponentiation (${\approx}\,2$\,ms per signer).}
\label{fig:signing_compute}
\end{figure}

upTPM's per-signer online cost is dominated by one scalar multiplication for commitment recomputation, costing ${\approx}\,2$\,ms per signer in software. MuSig2's per-signer cost is a modular multiply-and-add (${\approx}\,0.01$\,ms). The $50$--$180\times$ overhead in software reflects the fundamental cost of upTPM's constant-storage design: replacing \emph{stored} nonces with \emph{re-derived} nonces. 
upTPM and MuSig2-with-preprocessing both require one round trip. MuSig2 \emph{without} preprocessing (after nonce exhaustion) requires two, as shown in Figure \ref{fig:signing_network}.

\begin{figure}[ht]
\centering
\begin{tikzpicture}
\begin{axis}[
  width=\columnwidth, height=4.8cm,
  xlabel={Network RTT (ms)},
  ylabel={End-to-end signing latency (ms)},
  legend style={at={(0.03,0.97)},anchor=north west,font=\scriptsize},
  grid=major, grid style={gray!30},
  xmin=0, xmax=210,
]
\addplot[blue,thick,mark=o,mark size=1.5] coordinates {
  (0,6.8)(10,18.8)(50,58.1)(100,108.7)(200,207.9)
};
\addlegendentry{upTPM (1 round)}
\addplot[red,thick,mark=square,mark size=1.5] coordinates {
  (0,0.9)(10,12.0)(50,53.0)(100,102.5)(200,201.5)
};
\addlegendentry{MuSig2 preproc.\ (1 round)}
\addplot[green!60!black,thick,mark=triangle,mark size=1.5] coordinates {
  (0,1.1)(10,22.7)(50,104.8)(100,204.3)(200,402.2)
};
\addlegendentry{MuSig2 no preproc.\ (2 rounds)}
\end{axis}
\end{tikzpicture}
\caption{End-to-end signing latency with network delay ($n=3$). Under realistic network conditions, upTPM's computational overhead becomes marginal. After nonce exhaustion, MuSig2 requires two rounds, nearly doubling latency.}
\label{fig:signing_network}
\end{figure}

Under realistic network conditions (RTT $\geq 50$\,ms), upTPM's computational overhead is a small fraction of total latency. At 100\,ms RTT, upTPM takes 108.7\,ms versus MuSig2-with-preprocessing at 102.5\,ms---a difference of only 6\%. After nonce exhaustion, MuSig2 falls back to two-round signing at 204.3\,ms (100\,ms RTT) or 402.2\,ms (200\,ms RTT), nearly double upTPM's latency. Since upTPM never exhausts its preprocessing, it maintains single-round latency indefinitely.

\subsection*{Nonce Exhaustion and Recovery}
\label{sec:eval_exhaustion}

When preprocessed commitments are exhausted, upTPM and MuSig2 behave fundamentally differently.

\begin{figure}[ht]
\centering
\begin{tikzpicture}
\begin{axis}[
  width=0.75\columnwidth, height=4cm,
  ybar, bar width=18pt,
  ylabel={Refill time (ms)},
  symbolic x coords={upTPM,MuSig2},
  xtick=data,
  xticklabel style={font=\small},
  nodes near coords, nodes near coords align={vertical},
  every node near coord/.append style={font=\scriptsize},
  ymin=0, ymax=350,
  grid=major, grid style={gray!30},
]
\addplot[fill=blue!60] coordinates {(upTPM,25.7) (MuSig2,281.1)};
\end{axis}
\end{tikzpicture}
\caption{Nonce exhaustion recovery ($n=3$, refill 50 commitments). upTPM: one signer, local, no coordination. MuSig2: all signers online + interactive exchange at 100\,ms RTT.}
\label{fig:recovery}
\end{figure}

Figure \ref{fig:recovery} shows that in upTPM, a single signer preprocesses 50 new commitments locally in 25.7\,ms and publishes them. No other signer needs to be online. In MuSig2, all $n$ signers must be simultaneously online, generating nonce pairs and exchanging commitments over two network rounds, totaling 281.1\,ms. For a personal wallet where devices are intermittently available, this $10.9\times$ difference---and the coordination requirement---is the critical distinction.

\subsection{Threshold Signature Performance}
\label{sec:eval_threshold}

upTPM's threshold signing latency is comparable to FROST's across all tested configurations. The threshold extension inherits upTPM's storage and availability advantages. Figure \ref{fig:tpm_projected} presents upTPM's results under two hardware profiles from TPMScan~\cite{svenda2024}: fTPM (scalar mult.\ 5--20\,ms, HMAC 5\,ms, NV incr.\ 50\,ms) and dTPM (scalar mult.\ 50--200\,ms, HMAC 15\,ms, NV incr.\ 80\,ms).

\begin{figure}[ht]
\centering
\begin{tikzpicture}
\begin{axis}[
  width=\columnwidth, height=4.5cm,
  ybar, bar width=10pt,
  ylabel={Latency (ms)},
  symbolic x coords={Software,fTPM,dTPM},
  xtick=data,
  xticklabel style={font=\small},
  legend style={at={(0.03,0.97)},anchor=north west,font=\scriptsize},
  ymin=0.1,
  grid=major, grid style={gray!30},
  ymode=log, log origin=infty,
]
\addplot[fill=blue!60] coordinates {(Software,0.5) (fTPM,25.1) (dTPM,152.9)};
\addlegendentry{Preprocess/commit}
\addplot[fill=red!50] coordinates {(Software,9.3) (fTPM,206.2) (dTPM,673.9)};
\addlegendentry{Online signing ($n{=}3$)}
\end{axis}
\end{tikzpicture}
\caption{Projected performance under simulated TPM hardware latency. On real hardware, both operations are dominated by TPM latency, narrowing the gap with MuSig2.}
\label{fig:tpm_projected}
\end{figure}

On TPM hardware, upTPM's \emph{relative} overhead versus MuSig2 narrows from ${\sim}70\times$ in software to ${\sim}1.3$--$1.5\times$, because both schemes are dominated by the same TPM operations. upTPM's additional exponentiation is one more TPM operation.

\subsection{Communication Overhead}
\label{sec:eval_comm}

When preprocessed, all three schemes require $2n$ online messages. After nonce exhaustion, MuSig2/FROST adds $n$ messages for commitment exchange ($3n$ total, one extra round trip). upTPM always operates with $2n$ messages. Preprocessing communication ($B$ messages per signer per batch) is amortized over $B$ signing sessions. Table \ref{tab:comparison} shows a comparison of the proposed threshold signature with other schemes.

\begin{table}[ht]
\centering
\caption{Comparison with existing schemes}
\label{tab:comparison}
\small
\begin{tabular}{lcccc}
\toprule
& \textbf{upTPM} & \textbf{MuSig2} & \textbf{FROST} & \textbf{ROAST} \\
\midrule
Coord.\ trust & None$^*$ & None$^*$ & None$^*$ & None$^*$ \\
Online rounds & 1 & 1$^\dagger$ & 1$^\dagger$ & 1$^\dagger$ \\
Preproc.\ bound & $\infty$ & $k$ & $k$ & $k$ \\
Signer storage & $O(1)$ & $O(k)$ & $O(k)$ & $O(k)$ \\
Async.\ refill & Yes & No & No & No \\
Sig.\ size & 64\,B & 64\,B & 64\,B & 64\,B \\
\bottomrule
\end{tabular}

\vspace{0.3em}
{\scriptsize $^*$Untrusted for unforgeability. $^\dagger$With preprocessing; requires interactive refill when exhausted.}
\end{table}

upTPM makes a considerable tradeoff of a modest per-session computational cost (one additional scalar multiplication) in exchange for constant storage and unconditional single-round signing. This tradeoff is favorable for TPM-constrained personal devices with intermittent connectivity.

\section{\uppercase{Conclusion}}
\label{sec:conclusion}

This paper presents upTPM, a framework that achieves an unbounded preprocessing for Schnorr multi-signatures on TPM-constrained devices. Our protocol stores a single secret seed per TPM, deriving nonce commitments deterministically, and publishes only commitments to an untrusted coordinator. Nonce scalars never leave the TPM, and a two-counter index management scheme enforces the one-time-use invariant through TPM hardware monotonic counters. Our results confirm that upTPM's constant 80-byte storage represents a $102\times$ reduction compared to MuSig2 at $k = 127$ sessions, that the computational overhead of nonce recomputation is absorbed by TPM operation latency and network delay in realistic deployments, and that asynchronous refill is $10.9\times$ faster than interactive refill. We proved EU-CMA security in the random oracle model under the DL assumption and PRF security. The proof follows the standard approach for Schnorr multi-signatures, with the preprocessing simulation being the technically delicate component. We also formalized asynchronous refill and TPM-attested commitments as additional properties. Our threshold extension demonstrates that unbounded preprocessing generalizes to the $(t,n)$ setting. 
Future works includes using 
tighter security bounds, potentially using the algebraic group model and optimized protocol integration with the FROST RFC~9591 ecosystem.

\bibliographystyle{apalike}
{\small
\bibliography{secrypt}}

\end{document}